\documentclass[twocolumn]{IEEEtran}

\usepackage{amsmath}
\usepackage{amssymb}
\usepackage{revsymb4-1}
\usepackage{amsfonts}
\usepackage{graphics}
\usepackage{pgf}
\usepackage{lipsum}
\usepackage{enumitem}

\usepackage{theorem}
\usepackage{latexsym}

\newtheorem{theorem}{Theorem}

\newtheorem{corollary}[theorem]{Corollary}

\newtheorem{definition}[theorem]{Definition}

\newtheorem{lemma}[theorem]{Lemma}

\newtheorem{proposition}[theorem]{Proposition}

\newlength{\blank}
\newenvironment{remark}{\noindent\textbf{Remark\ }}{}

\mathchardef\ordinarycolon\mathcode`\:
\mathcode`\:=\string"8000
\def\vcentcolon{\mathrel{\mathop\ordinarycolon}}
\begingroup \catcode`\:=\active
  \lowercase{\endgroup
  \let :\vcentcolon}

\newcommand{\nc}{\newcommand}
\nc{\rnc}{\renewcommand}
\nc{\beq}{\begin{equation}}
\nc{\eeq}{{\end{equation}}}
\nc{\beqa}{\begin{eqnarray}}
\nc{\eeqa}{\end{eqnarray}}
\nc{\lbar}[1]{\overline{#1}}
\nc{\ket}[1]{|#1\rangle}
\nc{\bra}[1]{\langle#1|}
\nc{\braket}[2]{\langle #1 | #2 \rangle}
\nc{\ketbra}[2]{|#1\rangle\!\langle#2|}
\nc{\proj}[1]{| #1\rangle\!\langle #1 |}
\nc{\avg}[1]{\langle#1\rangle}
\nc{\Rank}{\operatorname{Rank}}
\nc{\smfrac}[2]{\mbox{$\frac{#1}{#2}$}}
\nc{\tr}{\operatorname{Tr}}
\nc{\ox}{\otimes}
\nc{\dg}{\dagger}
\nc{\dn}{\downarrow}
\nc{\cA}{\mathcal{A}}
\nc{\cB}{\mathcal{B}}
\nc{\cC}{\mathcal{C}}
\nc{\cD}{\mathcal{D}}
\nc{\cE}{\mathcal{E}}
\nc{\cF}{\mathcal{F}}
\nc{\cG}{\mathcal{G}}
\nc{\cH}{\mathcal{H}}
\nc{\cI}{\mathcal{I}}
\nc{\cJ}{\mathcal{J}}
\nc{\cK}{\mathcal{K}}
\nc{\cL}{\mathcal{L}}
\nc{\cM}{\mathcal{M}}
\nc{\cN}{\mathcal{N}}
\nc{\cO}{\mathcal{O}}
\nc{\cP}{\mathcal{P}}
\nc{\cR}{\mathcal{R}}
\nc{\cS}{\mathcal{S}}
\nc{\cT}{\mathcal{T}}
\nc{\cX}{\mathcal{X}}
\nc{\cZ}{\mathcal{Z}}
\nc{\csupp}{{\operatorname{csupp}}}
\nc{\qsupp}{{\operatorname{qsupp}}}
\nc{\var}{\operatorname{var}}
\nc{\rar}{\rightarrow}
\nc{\lrar}{\longrightarrow}
\nc{\polylog}{\operatorname{polylog}}
\nc{\1}{{\openone}}
\nc{\id}{{\operatorname{id}}}

\nc{\RR}{{{\mathbb R}}}
\nc{\CC}{{{\mathbb C}}}
\nc{\FF}{{{\mathbb F}}}
\nc{\NN}{{{\mathbb N}}}
\nc{\ZZ}{{{\mathbb Z}}}
\nc{\PP}{{{\mathbb P}}}
\nc{\QQ}{{{\mathbb Q}}}
\nc{\UU}{{{\mathbb U}}}
\nc{\EE}{{{\mathbb E}}}

\nc{\qed}{{$\hfill\Box$}}

\def\hookto{\hookrightarrow}
\def\>{\rangle}
\def\<{\langle}
\def\be{\begin{equation}}
\def\ee{\end{equation}}
\def\bea{\begin{eqnarray}}
\def\eea{\end{eqnarray}}

\begin{document}

\title{Potential Capacities of Quantum Channels}

\author{Andreas Winter\thanks{Andreas Winter is with ICREA and
F\'{\i}sica Te\`{o}rica: Informaci\'{o} i Fen\`{o}mens Qu\`{a}ntics
Universitat Aut\`{o}noma de Barcelona,
ES-08193 Bellaterra (Barcelona), Spain. Email: andreas.winter@uab.cat.}
and Dong Yang\thanks{Dong Yang is with
F\'{\i}sica Te\`{o}rica: Informaci\'{o} i Fen\`{o}mens Qu\`{a}ntics
Universitat Aut\`{o}noma de Barcelona,
ES-08193 Bellaterra (Barcelona), Spain and
Laboratory for Quantum Information,
China Jiliang University,
Hangzhou, Zhejiang 310018, China. Email: dyang@cjlu.edu.cn.}
\thanks{Manuscript date 16 January 2016.}
}

\markboth{Winter and Yang: Potential Capacities of Quantum Channels}{Winter and Yang: Potential Capacities of Quantum Channels}
\maketitle

\begin{abstract}
We introduce \emph{potential capacities} of quantum channels in 
an operational way and provide upper bounds for these quantities, which
quantify the ultimate limit of usefulness of a channel for a given task
in the best possible context. 

Unfortunately, except for a few isolated cases, potential capacities
seem to be as hard to compute as their ``plain'' analogues. We thus
study upper bounds on some potential capacities:
For the classical capacity, we give an upper bound in terms of 
the entanglement of formation. 
To establish a bound for the quantum and private capacity, we first ``lift'' 
the channel to a Hadamard channel and then prove that the quantum and 
private capacity of a Hadamard channel is strongly additive, implying that 
for these channels, potential and plain capacity are equal. 
Employing these upper bounds we show that if a channel is noisy, however 
close it is to the noiseless channel, then it cannot be activated into 
the noiseless channel by any other contextual channel; this 
conclusion holds for all the three capacities. We also discuss the so-called
environment-assisted quantum capacity, because we are able to
characterize its ``potential'' version.
\end{abstract}


\IEEEpeerreviewmaketitle


\section{Motivation}
\label{sec:intro}
\IEEEPARstart{T}{he} central problem in information theory is to find the capacity of a noisy 
channel for transmitting information faithfully. Depending on what type of information is to be sent, there are several capacities that can be defined for a quantum 
channel, among them the classical capacity \cite{Holevo,SW}, the quantum 
capacity \cite{Lloyd-Q,Shor,Devetak-PQ} and the private capacity \cite{Devetak-PQ,CaiWinterYeung-P}. 
In contrast to classical information theory, where the capacity is expressed 
by Shannon's famous single-letter formula, the status of quantum channel capacities
is much more complicated. 
The relevant quantities are known to be 
non-additive \cite{Hastings,DiVincenzo-Qnonadditive,Li,Smith-privacy}, which is at 
the center of interest in quantum information science, and the best known formula 
to calculate the capacities involves optimization over growing numbers of channel 
uses (``regularization''), where we have to perform an optimization over an infinite 
number variables, making a head-on numerical approach impossible,
cf.~\cite{Cubitt-et-al:reg,ES-reg}.
This makes it difficult to answer questions related to capacities, even some 
simple qualitative ones, such as whether, given a quantum channel, it is
useful to transmit quantum information.
Non-additivity in quantum Shannon theory is due to entanglement,
which has no classical counterpart. 
Employing entangled inputs for the channels, it is possible to transmit more 
information than just using product inputs. Entangled inputs between different 
quantum channels open the door to all kinds of effects that are
impossible in classical information theory. An extreme phenomenon is 
superactivation \cite{SmithYard}; there exist two quantum channels that cannot 
transmit quantum information when they are used individually, but can transmit
at positive rate when they are used together.

The phenomenon of superactivation, and more broadly of
super-additivity, implies that the capacity of a quantum channel does not adequately characterize the channel, since the utility of 
the channel depends on what other contextual channels are available. 
So it is natural to ask the following question: 
What is the maximum possible capability of a channel to transmit 
information when it is used in combination with any other contextual 
channels? We introduce the potential capacity to formally capture
this notion.

Superactivation can also be rephrased in an alternative way, that a 
zero-capacity channel becomes a positive-capacity one under the assistance of 
another zero-capacity side channel. Superactivation of quantum channel capacity shows that entangled inputs across different channel uses can provide a dramatic advantage, but more generally entangling different channel uses can give rise to superadditivity, i.e., an increase of the capacity above the sum of the channel capacities when the two channels are used jointly. Superactivation exhibits one regime of entanglement advantage, the regime of low
capacity. Could entanglement help in this sense at the other extreme? That is 
``Can a noisy channel, whose quantum capacity 
is $\leq \log d -\delta$, become perfectly noiseless under the assistance 
of a suitable zero-capacity side channel?'' 
Since it is difficult to characterize all the zero-capacity channels, it 
seems hard to answer this question. Encouraged by superactivation, 
one might guess that a noisy channel could behave like a noiseless channel 
by the assistance of a proper zero-capacity side channel. In this work, we 
will provide upper bounds on the potential capacities to exclude this 
possibility. In this sense, entanglement can help but cannot help too much.

This paper is structured as follows. 
In Section \ref{sec:defs} we introduce notation, definitions and state some 
basic known facts. 
In particular, we review the regularized formulas of three capacities (classical, 
quantum, and private capacity), and the results about additivity of degradable channels, 
furthermore the entanglement-assisted and the environment-assisted capacities.
In Section \ref{sec:potcap}, we introduce the notion of potential capacity 
and in Section \ref{sec:cases} evaluate it or give upper bounds for it, and 
prove that an imperfect channel cannot be activated into a perfect one. 
Finally we end with a summary and open questions in Section \ref{sec:open}.


\section{Notation and preliminaries}
\label{sec:defs}
We assume that all Hilbert spaces, denoted ${\cal H}$, 
are finite dimensional. 
Recall that a quantum state $\rho$ is a linear operator on ${\cal H}$ satisfying 
$\rho\ge 0$ and $\tr\rho=1$. A quantum channel is a completely positive and trace 
preserving (CPTP) linear map, from an input system $A$ to output system $B$ (we shall
generally use the same names for the underlying Hilbert spaces). 
From the Stinespring dilation theorem \cite{Stinespring}, we know that for a channel ${\cal N}$, there 
always exist an isometry $U:A \hookrightarrow B\otimes E$ for some environment 
space $E$, i.e.~$U^\dagger U = \1$, such that 
${\cal N}(\rho)=\tr_E U \rho U^\dagger$. 
The complementary channel of ${\cal N}$, which we denote ${\cal N}^c$, is the channel 
that maps from the input system $A$ to the environment system $E$, obtained by taking
the partial trace over system $B$ rather than the environment: 
${\cal N}^c(\rho)=\tr_B U \rho U^\dagger$. 
Since the Stinespring dilation is unique up to a change of basis of environment 
$E$, ${\cal N}^c$ is well-defined up to unitary operations on $E$.  
A quantum channel has another representation known as the Kraus representation: 
${\cal N}(\rho)=\sum_iK_i\rho K_i^{\dagger}$, where $K_i$ are called Kraus operators 
satisfying $\sum_iK_i^{\dagger}K_i=\1$.
Given a multipartite state $\rho^{ABC}$, we write $\rho^A = \tr_{BC} \rho^{ABC}$ for 
the corresponding reduced state. The von Neumann entropy is defined as 
$S(A)_\rho = S({\rho^A}) =-\tr\rho^A \log\rho^A$. The conditional von Neumann 
entropy of $A$ given $B$ is defined as $S(A|B)=S({\rho^{AB}})-S({\rho^B})$, 
the mutual information $I(A:B)=S({\rho^A})+S({\rho^B})-S({\rho^{AB}})$, 
and the conditional mutual information as
$I(A:B|C)=S(\rho^{AC})+S(\rho^{BC})-S({\rho^{ABC}})-S(\rho^C)$. 
When there is no ambiguity as to which state is being referred, we simply write 
$S(A)=S(\rho^A)$.

We review the regularization formulas of the three principal capacities: 
the classical, quantum, and private capacity.
 
The \emph{classical capacity} of a quantum channel is the rate at which one can 
reliably send classical information through a quantum channel, and is 
given by \cite{Holevo,SW},
\be
  C({\cN})=\lim_{n\to\infty} \frac1n \chi\left( {\cN}^{\otimes n} \right),
\ee
with the \emph{Holevo capacity} $\chi({\cal N})$ defined as
\be
  \chi({\cal N})=\max_{\{p_i,\phi_i\}} S\left(\sum_i p_i{\cal N}(\phi_i)\right)
                                       -\sum_i p_i S\bigl({\cal N}(\phi_i)\bigr).
\ee
Note the elementary rewriting of the Holevo capacity as follows,
known as the \emph{MSW identity} \cite{MSW}:
\be
  \label{eq:MSW}
  \chi({\cal N})=\max_{\rho^A} S(\rho^B)-E_F(\rho^{BE}),
\ee
where $\rho^{BE}=U\rho^A U^{\dagger}$, $U$ is the Stinespring 
isometry of ${\cal N}$, and $E_F(\rho^{BE})$ is the 
\emph{entanglement of formation} of the bipartite state $\rho^{BE}$ defined as
\be
  E_F(\rho^{BE})=\min \sum_i p_iS(\phi_i^B), \text{ s.t. } \rho^{BE} = \sum_i p_i \proj{\phi_i}^{BE}.
\ee

The \emph{quantum capacity} of a quantum channel is 
the rate at which one can reliably send quantum information 
through a quantum channel, and is given by \cite{Lloyd-Q,Shor,Devetak-PQ}, 
\be
Q({\cal N})=\lim_{n\to\infty} \frac1n Q^{(1)}({\cal N}^{\otimes n}),
\ee
with $Q^{(1)}({\cal N})$ defined as
\begin{equation}\begin{split}
Q^{(1)}({\cal N}) &= \max_{\ket{\phi}^{RA}} S\bigl({\cal N}(\phi^A)\bigr)
                                             -S\bigl(\id\otimes{\cal N}(\phi^{RA})\bigr),\\
                  &= \max_{\rho^A} (S(\rho^{B})-S(\rho^{E})),
\end{split}\end{equation}
where $\rho^{BE}=U\rho^A U^{\dagger}$, $U$ is the isometry of ${\cal N}$, and $S(\rho^{B})-S(\rho^{E})=S(\rho^{B})-S(\rho^{RB})$ is known as coherent information \cite{Schumacher,SN,BNS,BKN}.

The \emph{private capacity} of the quantum channel ${\cal N}$ is 
given by \cite{Devetak-PQ,CaiWinterYeung-P}
\be
  P({\cal N})=\lim_{n\to\infty} \frac1n P^{(1)}({\cal N}^{\otimes n}),
\ee
with $P^{(1)}({\cal N})$ defined as
\begin{align*}
  P^{(1)}({\cal N}) &=\max_{\{p_t,\rho_t\}} (I(T:B)-I(T:E)), \\ \text{ with respect to } 
  \rho^{TBE}        &=\sum_{t} p_t \proj{t}^T \otimes U \rho_t^A U^{\dagger}.
\end{align*}

Now we recall the definition of a degradable channel and its properties on 
quantum and private capacities.
\begin{definition}
  A channel ${\cal N}$ is called \emph{degradable}~\cite{Devetak-degradable} if
  it can simulate its complementary channel ${\cal N}^c$, i.e.~there is a 
  degrading CPTP map ${\cal D}$ such that ${\cal D} \circ {\cal N} = {\cal N}^c$.
\end{definition}

\begin{lemma}[Devetak/Shor~\cite{Devetak-degradable}] 
  If ${\cal N}$ and ${\cal M}$ are degradable channels, then their single-letter 
  quantum capacity is additive: 
  $Q^{(1)}({\cal N}\otimes{\cal M})=Q^{(1)}({\cal N})+Q^{(1)}({\cal M})$.
\end{lemma}

\begin{lemma}[Smith~\cite{Smith-degradable-p}] 
  If a quantum channel ${\cal N}$ is a degradable channel, then its quantum
  capacity is equal to its private capacity, and both are given by the single-letter
  coherent information:
  $Q({\cal N})=P({\cal N})=Q^{(1)}({\cal N})=P^{(1)}({\cal N})$.
\end{lemma}

\medskip
We furthermore recall two other capacities: The 
\emph{entanglement-assisted classical capacity} of ${\cal N}$ \cite{BSST}, which is the capacity for transmitting classical information through the channel with the help of unlimited prior entanglement shared between the sender and the receiver and 
which is given by the simple and beautiful formula
\be
  C_E(\cN) = \max_{\rho^A} I(R:B),
\ee
where $I(R:B) = S(\rho^R)+S(\rho^B)-S(\rho^{RB})$ is the quantum mutual
information of the state $\rho^{RB} = (\id\otimes{\cal N})(\proj{\phi}^{RA})$,
with a purification $\ket{\phi}^{RA}$ of $\rho^A$.
And the \emph{environment-assisted quantum capacity}, which refers to
active feed-forward of classical information from the channel environment
$E$ to the receiver $B$ \cite{SVW,Winter-envass}, is given by
\be
  Q_A({\cal N}) = \max_{\rho^A} \min\Bigl\{ S(\rho^A),S\bigl({\cal N}(\rho^A)\bigr) \Bigr\}.
\ee


\section{Potential capacities}
\label{sec:potcap}
Notice that the formulas for $C$, $P$ and $Q$ all are regularized expressions
due to the non-additivity of their respective single-letter quantities,
$\chi$, $P^{(1)}$ and $Q^{(1)}$.  

We call a real function $f({\cal N})$ on the set of channels \emph{weakly-additive} 
if $f({\cal N}^{\otimes n}) = nf({\cal N})$ for all $n\ge 1$, 
and \emph{strongly-additive} if $f({\cal N}\otimes {\cal M})= f({\cal N})+f({\cal M})$ 
for any channels ${\cal N}$ and ${\cal M}$. 
Obviously, if $f$ is strongly-additive, then it is also weakly-additive but not 
vice versa; and example of this is given by the environment-assisted capacity
$Q_A({\cal N})$. 
Furthermore, for fixed $f$, we call a channel ${\cal N}$ \emph{strongly-additive},
if for all other channels ${\cal M}$,
$f({\cal N}\otimes {\cal M})= f({\cal N})+f({\cal M})$.

From their expression as regularizations, or directly from the definition,
one can directly deduce that the capacities $C({\cal N})$, $Q({\cal N})$ and 
$P({\cal N})$ are weakly-additive. 
Furthermore, it is known that neither $Q({\cal N})$ nor $P({\cal N})$ are 
strongly-additive; $C({\cal N})$ is believed to be not strongly-additive, though this has
not been proved so far. The single-letter quantities $\chi({\cal N})$, 
$Q^{(1)}({\cal N})$, $P^{(1)}({\cal N})$ are not even weakly-additive. 

Due to their non-additivity, the capability to transmit information through a quantum channel does not only depend on the channel itself, but also on any contextual channel with which it can be combined. So the standard capacity cannot uniquely characterize the utility of the channel. It is natural to consider the maximal possible capability to transmit information when it is used in combination with any other contextual channels. We introduce the potential capacity to describe this notion. It describes the potential capability that can be activated by a proper contextual channel. Since the three capacities share the same property, we define the notion in a unified way.

In the following definitions, we assume a super-additive function $f$,
i.e.~$f({\cal N}\otimes {\cal M}) \geq f({\cal N})+f({\cal M})$ 
for any channels ${\cal N}$ and ${\cal M}$, so that the regularization
$f^{(\infty)}$ is given by
\be
\label{eq:regularization}
  f^{(\infty)}({\cal N}) = \sup_n \frac1n f({\cal N}^{\otimes n})
                         = \lim_{n\rightarrow\infty} \frac1n f({\cal N}^{\otimes n}).
\ee
By its definition, $f^{(\infty)}$ is always weakly-additive, and 
$f({\cal N}) \leq f^{(\infty)}({\cal N})$.

\begin{definition}
\label{def:f-pot}
For a channel ${\cal N}$, the \emph{potential capacity} associated to $f$
is defined as 
\be
\label{eq:pot}
  f^{(\infty)}_p({\cal N}) 
     := \sup_{\cal M} \left[ f^{(\infty)}({\cal N}\otimes{\cal M})-f^{(\infty)}({\cal M}) \right],
\ee
where $f^{(\infty)}({\cal N})$ is the regularization of $f$.

Similarly, the \emph{potential single-letter capacity} is defined as
\be
\label{eq:one-pot}
  f^{(1)}_p({\cal N}) := \sup_{\cal M} \left[ f({\cal N}\otimes{\cal M})-f({\cal M}) \right],
\ee
where $f^{(1)}({\cal N}) = f({\cal N})$ is the single-letter function.

This notion has been introduced before in~\cite[Sec.~VII]{Smith}, for 
the case of $f = Q^{(1)}$, under the name of ``quantum value added capacity.''
\end{definition}

\medskip
Note that we have (always assuming super-additivity of $f$)
\[
  f^{(1)}_p({\cal N}) = f({\cal N}) \text{ iff } {\cal N} \text{ is strongly additive.}
\]

Eq. (\ref{eq:pot}) is difficult to calculate because of the unlimited dimension of the contextual channel $\cM$ and the regularization function $f^{\infty}$. Eq. (\ref{eq:one-pot}) looks simpler but still suffers from the unlimited dimension problem. So we would like to provide upper bounds for them. Before that we will show that the notion ``potential'' is intrinsically sub-additive and a little surprising fact: $f^{(\infty)}_p({\cal N})\le f^{(1)}_p({\cal N})$ though $f^{(\infty)}({\cal N})\ge f^{(1)}({\cal N})$.
 
\begin{lemma} 
  \label{lemma:pot-subadd}
  For any super-additive $f$, both $f^{(1)}_p({\cal N})$ and $f_p^{(\infty)}({\cal N})$ are
  sub-additive, i.e.
  \begin{align*}
    f^{(1)}_p({\cal N}\otimes{\cal M})      &\le f^{(1)}_p({\cal N})+f^{(1)}_p({\cal M}), \\
    f^{(\infty)}_p({\cal N}\otimes{\cal M}) &\le f^{(\infty)}_p({\cal N})+f^{(\infty)}_p({\cal M}).
  \end{align*}
\end{lemma}

\begin{proof}
We prove the claim for $f^{(1)}_p({\cal N})$; the proof for 
$f^{(\infty)}_p({\cal N})$ is similar. Namely, for an arbitrary channel ${\cal T}$,
\[\begin{split}
 &\quad~ f^{(1)}({\cal N}\otimes{\cal M}\otimes{\cal T})-f^{(1)}({\cal T}) \\
    &=f^{(1)}({\cal N}\otimes{\cal M}\otimes{\cal T})-f^{(1)}({\cal M}\otimes{\cal T})\\
                                         & \quad~+f^{(1)}({\cal M}\otimes{\cal T})-f^{(1)}({\cal T}),\\
    &\le \sup_{\cal S} \left[f^{(1)}({\cal N}\otimes{\cal S})-f^{(1)}({\cal S}) \right]\\
         &\quad~+\sup_{\cal S} \left[f^{(1)}({\cal M}\otimes{\cal S})-f^{(1)}({\cal S}) \right], \\
    &=   f^{(1)}_p({\cal N})+f^{(1)}_p({\cal M}).
\end{split}\]
Maximization over ${\cal T}$ concludes the proof.
\end{proof}

\begin{lemma}
  \label{lemma:chain}
  The potential capacity is upper bounded by the potential single-letter capacity, 
  more precisely
  \[
    f^{(1)}({\cal N}) \le f^{(\infty)}({\cal N}) \le f^{(\infty)}_p({\cal N}) \le f^{(1)}_p({\cal N}).
  \]
\end{lemma}

\begin{proof} 
The first ``$\le$'' comes from Eq. (\ref{eq:regularization}) by taking $n=1$ and the second ``$\le$'' from Eq. (\ref{eq:pot}) by taking $\cM$ as a fixed state channel.  
For the third ``$\le$'', consider the following chain of inequalities:
\[\begin{split}
 &\quad~f^{(\infty)}({\cal N}\otimes{\cal M})-f^{(\infty)}({\cal M})\\
      &=   \lim_{n\to\infty} \frac1n f({\cal N}^{\otimes n}\otimes{\cal M}^{\otimes n})
              -\lim_{n\to\infty} \frac1n f({\cal M}^{\otimes n}),                      \\
      &=   \lim_{n\to\infty} \frac1n \left[ f({\cal N}^{\otimes n}\otimes{\cal M}^{\otimes n})
                                             -f({\cal M}^{\otimes n}) \right],         \\
      &\le \lim_{n\to\infty} \frac1n f^{(1)}_p({\cal N}^{\otimes n}),                  \\
      &\le \lim_{n\to\infty} \frac1n n f^{(1)}_p({\cal N})
       =    f^{(1)}_p({\cal N}),
\end{split}\]
where the first inequality uses the definition of the potential single-shot 
capacity and the second one the sub-additivity.

Hence we have
\[
  f^{(1)}({\cal N}) \le f^{(\infty)}({\cal N}) \le f^{(\infty)}_p({\cal N}) \le f^{(1)}_p({\cal N}).
\]
\end{proof}

\medskip
\begin{remark} 
\normalfont
Notice that all capacities and their single-letter formulations are
super-additive, and that the single-letter form is a lower bound of 
the regularized form. However, their ``potential'' counterparts have the 
reverse relation; this was glimpsed in \cite{Smith} without any further 
investigation.
\end{remark}


\section{Five concrete potential capacities}
\label{sec:cases}
Now we can turn to five concrete examples. We start with
the entanglement-assisted capacity, which presents a trivial case:
Namely, $C_E$ is known to be strongly-additive \cite{AdamiCerf}, i.e., for all channels
${\cal N}$ and ${\cal M}$, 
$C_E({\cal N}\otimes{\cal M}) = C_E({\cal N})+C_E({\cal M})$.
Thus, $C_E$ equals its own regularization and in turn its own
potential capacity:
\[
  C_E({\cal N}) = C_E^{(\infty)}({\cal N}) = (C_E)_p({\cal N}).
\]

The next subsection presents the slightly more interesting
case of $Q_A$, which is not additive, but it has a single-letter formula. 
For this case we are still able to evaluate $(Q_A)_p({\cal N})$
in a simple single-letter formula, but for the subsequent $C$, $P$ and $Q$ we will
only be able to give upper bounds.

\subsection{Potential environment-assisted capacity}
\label{subsec:Q_A}
There are two types of channels ${\cal M}$ with $Q_A({\cal M}) = 0$, 
which we will use to activate a given ${\cal N}$, on one hand,
those with one-dimensional input system, on the other those with
one-dimensional output system. Their Stinespring isometries are
\begin{align*}
  &V_1 : \mathbb{C} \longrightarrow B' \otimes E', \quad~~1 \longmapsto     \ket{\phi}^{B'E'},\\
  &V_2 : A'  \longrightarrow \mathbb{C} \otimes E', \quad \ket{\psi} \longmapsto     1^{B'} \otimes (W_2\ket{\psi})^{E'},
\end{align*}
where $W_2$ is an isometry. Using these, we show the following simple result:

\begin{theorem}
For any channel ${\cal N}: A\rightarrow B$,
\[
\begin{split}
  (Q_A)_p({\cal N}) &= \max_{\rho^A} \max \Big\{ S(\rho^A), S\bigl({\cal N}(\rho^A)\bigr) \Bigr\},\\
                    &= \max \left\{ \log|A|, \max_{\rho^A} S\bigl({\cal N}(\rho^A)\bigr) \right\}.
\end{split}
\]
\end{theorem}

\begin{proof}
First, for ``$\geq$'': By tensoring with a channel ${\cal M}: A'\rightarrow B'$
 of the above type having zero environment-assisted capacity, i.e.
either ${\cal M}_1$ where the only input state has zero
entropy, or ${\cal M}_2$ where the only output state has
zero entropy. In this way we can bump up either the output
entropy $S\bigl({\cal N}\otimes{\cal M}(\rho^{AA'})\bigr)$, 
or the input entropy $S(\rho^{AA'})$ by an arbitrary amount,
without changing the respective other. Thus indeed,
\[
  (Q_A)_p({\cal N}) \geq Q_A({\cal N}\otimes{\cal M})
                    \geq \max \Big\{ S(\rho^A), S\bigl({\cal N}(\rho^A)\bigr) \Bigr\}.
\]

In the other direction, consider an arbitrary channel ${\cal M}$. Then we have,
\[\begin{split}
  (Q_A)_p({\cal N}) &= \sup_{{\cal M}} Q_A({\cal N}\otimes{\cal M}) - Q_A({\cal M}), \\
                    &\leq \sup_{{\cal M}} \max_{\rho^{AA'}} \left(
                          \min \Bigl\{ S(\rho^{AA'}),S\bigl({\cal N}\otimes{\cal M}(\rho^{AA'})\bigr) \Bigr\} \right. \\
                                     &~~~~~~~~~~~~~~~ \left.+ \max \Bigl\{ -S(\rho^{A'}),-S\bigl({\cal M}(\rho^{A'})\bigr) \Bigr\}
                                                                                                     \right),\\
                    &\leq \sup_{{\cal M}} \max_{\rho^{AA'}} 
                          \max \bigl\{ S(A|A'), S(B|B') \bigr\},                     \\
                    &\leq \sup_{{\cal M}} \max_{\rho^{A}} 
                          \max \bigl\{ S(A), S(B) \bigr\},
\end{split}\]
and we are done.
\end{proof}

\subsection{Potential classical capacity}
\label{subsec:C}
In this section, we study the potential classical capacity
and its relation to the single-letter Holevo capacity, and
most importantly
establish an upper bound via a specific entanglement measure.
This bound is used to prove that an imperfect quantum channel cannot be activated 
into a perfect one by any other contextual channel. 

\begin{definition} 
Specializing Definition~\ref{def:f-pot}
to the case $f\equiv C$, we obtain the \emph{potential classical capacity}
\be
  C_p({\cal N}) = \sup_{\cal M} \bigl[ C({\cal N}\otimes{\cal M})-C({\cal M}) \bigr],
\ee
and likewise the \emph{potential Holevo capacity}
\be
  \chi_p({\cal N}) = \sup_{\cal M} \bigl[ \chi({\cal N}\otimes{\cal M})-\chi({\cal M}) \bigr].
\ee
\end{definition}

By Lemma~\ref{lemma:chain}, we have
\be
  \chi({\cal N})\le C({\cal N}) \le C_p({\cal N}) \le \chi_p({\cal N}).
\ee
To give non-trivial bounds on $\chi_p({\cal N})$, we invoke the following
previous result.

\begin{lemma}[Yang \emph{et al.}~\cite{Yang}] 
\label{lemma:yang}
For a mixed four-partite state $\rho^{B_1B_2E_1E_2}$,
\be
  \label{se:main}
  E_F(\rho^{B_1B_2:E_1E_2}) \ge G(\rho^{B_1:E_1}) + E_F(\rho^{B_2:E_2}),
\ee
where the fuction $G(\rho^{BE})$ is defined as
\begin{align}
  G(\rho^{B:E}) := \min_{\{p_i,\rho^{BE}_i\}} \sum_i p_i C_{\leftarrow}(\rho^{BE}_i), \\ \text{with~}
  C_{\leftarrow}(\sigma^{BE}) = S(\sigma_B) - \min_{\{P_j\}} r_j S(\sigma^B_j),\nonumber
\end{align}
where $\{P_j\}$ ranges over POVMs on $E$, i.e.~$P_j\geq 0$ and $\sum_j P_j=\1$,
$r_j= \tr (\1\otimes P_j)\sigma^{BE}$, and 
$\sigma^B_j = \frac{1}{r_j} \tr_E (\1\otimes P_j)\sigma_{BE}$.

Furthermore, $G(\rho^{B:E})$ is \emph{faithful}, 
meaning $G(\rho^{B:E})=0$ iff $\rho^{BE}$ is separable.
\end{lemma}

\begin{theorem} 
For a channel ${\cal N}$ with Stinespring isometry $U$,
\be
  \chi_p({\cal N}) \le \max_{\rho^A} \bigl[ S(\rho^B)-G(\rho^{BE}) \bigr],
\ee
where $\rho^{BE} = U \rho^A U^{\dagger}$.
\end{theorem}

\begin{proof}
Using the MSW identity, Eq.~(\ref{eq:MSW}), and Lemma~\ref{lemma:yang}, 
we have the following chain of identities and inequalities:
\[\begin{split}
&\chi({{\cal N}}\otimes{\cal M})\\
          =  &\max_{\rho^{A_1A_2}} S(B_1B_2)-E_F(B_1B_2:E_1E_2),\\
         \le &\max_{\rho^{A_1A_2}} \Bigl\{ S(B_1)+S(B_2) - \bigl[ G(B_1:E_1)+E_F(B_2:E_2) \bigr] \Bigr\}, \\
         =   &\max_{\rho^{A_1A_2}} \Bigl\{ \bigl[ S(B_1)-G(B_1:E_1) \bigr]
                                         + \bigl[ S(B_2)-E_F(B_2:E_2) \bigr] \Bigr\},      \\
          \le &\max_{\rho^{A_1}}    \bigl[ S(B_1)-G(B_1:E_1) \bigr] 
                 +\max_{\rho^{A_2}} \bigl[ S(B_2)-E_F(B_2:E_2) \bigr],                     \\
          =  &\max_{\rho^{A_1}}    \bigl[ S(B_1)-G(B_1:E_1) \bigr] + \chi({\cal M}).
\end{split}\]
By definition of $\chi_p$, the claim follows.
\end{proof}

In \cite{Brandao}, a channel is perfect when its capacity is $\log d_{out}$. In the general case, the input space may have the different dimension from the output space. It is obvious that the capacity of the channel is upper-bounded by $\min\{\log d_{in}, \log d_{out}\}$. Here we call a channel perfect if its capacity is equal to $\log d_{\min}$ with $d_{\min}=\min\{d_{in}, d_{out}\}$ and we prove the following corollary.

\begin{corollary}
\label{cor:classical}
If a quantum channel ${\cal N}$ is not perfect for 
transmitting classical information in the single-letter sense, then 
it cannot be activated to the perfect one by any contextual channel: 
\[
  \chi({\cal N}) < \log d_{\min} \Longrightarrow \chi_p({\cal N}) < \log d_{\min}.
\]
\end{corollary}

\begin{proof} 
Suppose the potential capacity of the channel is $\chi_p({\cal N}) = \log d_{\min}$.

In the case of $d_{\min}=d_{out}=d$,  from $C_p({\cal N}) \le \chi_p({\cal N}) \le \max_{\rho_A}[S(\rho_B)-G(\rho_{BE})]$, 
we know that there is an input state $\rho^A$ such that for 
$\rho^{BE} = U \rho^A U^{\dagger}$, we have
$S(\rho_B)=\log d$ and $G(\rho_{BE})=0$. 

Since $G$ is faithful (see Lemma~\ref{lemma:yang}), this means that 
$\rho_{BE}$ is separable, which amounts to $E_F(\rho_{BE})=0$. 
From the MSW identity, Eq.~(\ref{eq:MSW}), we obtain that
$\chi({\cal N})=\log d$, which means the channel is perfect in the single-letter
sense.

In the case of $d_{\min}=d_{in}=d$, suppose $C_p({\cal N})=\log d=S(A)=S(BE)$, 
where $\rho^A=\frac{1}{d}\1$ and $\rho^{BE} = U \rho^A U^{\dagger}$. 
From Lemma~\ref{lemma:IGE} in the Appendix, we obtain 
$\log d=S(A)=C_p({\cal N})\le \max_{\rho_A}[S(\rho_B)-G(\rho_{B:E})]\le \max_{\rho_A} S(A)=\log d$. 
So $\rho^A=\frac{1}{d}\1$ is the optimal input to achieve $\max_{\rho_A}[S(\rho_B)-G(\rho_{BE})]$. 
This means $G(\rho_{B:E})=S(B)-S(BE)$ for the state $\rho^{BE}$. 
Also from Lemma~\ref{lemma:IGE}, we know that $E_F(B:E)=S(B)-S(BE)$,
meaning that $\chi({\cal N})=\log d$.
\end{proof}

\medskip
\begin{remark}
\normalfont
In \cite{Brandao}, it is shown that if $\chi(\cN)<\log d_{out}$, then $C(\cN)<\log d_{out}$. Notice that Holevo capacity is the capacity when the codewords are restricted to product states. That is to say if the capacity when using product state encoding cannot achieve the possibly maximal quantity $\log d_{out}$, then it cannot either when using entangled state encoding. In other words, an imperfect channel cannot be activated to a perfect one by itself. Corollary \ref{cor:classical} is stronger in two points. One is that it covers the case $d_{in}<d_{out}$ where \cite{Brandao} says nothing about. Indeed it is not immediately to obtain so we need the Appendix to deal with this case. The other point is that Corollary \ref{cor:classical} asserts an imperfect channel cannot be activated to a perfect one by {\it any} channel. The reasoning for $d_{\min}=d_{out}$ is almost the same as that 
in \cite{Brandao} but for $d_{\min}=d_{in}$ we need more. Here we emphasiz that we use the particular 
entanglement measure $G(\rho_{BE})$ while other entanglement measures may be employed to prove the result in \cite{Brandao}.
\end{remark}

\subsection{Potential quantum capacity}
\label{subsec:Q}
In this section, we move on to the potential quantum capacity and study its 
relations to the single-letter quantity $Q^{(1)}({\cal N})$. In~\cite[Sec.~VII]{Smith},
this had been introduced under the name of ``quantum value added capacity'',
and our Lemma~\ref{lemma:pot-subadd} already been observed in that case.
Here, we establish an upper bound in terms of the entanglement of formation
of the channel, and finally prove that an imperfect quantum channel cannot be 
activated into a perfect one by any other contextual channel. 

\begin{definition}
Specializing Definition~\ref{def:f-pot}
to the case $f\equiv Q$, we obtain the \emph{potential quantum capacity}
\be
  Q_p({\cal N}) = \sup_{\cal M} \bigl[ Q({\cal N}\otimes{\cal M})-Q({\cal M}) \bigr],
\ee
and the \emph{potential single-letter quantum capacity}
\be
  Q^{(1)}_p({\cal N}) = \sup_{\cal M} \bigl[ Q^{(1)}({\cal N}\otimes{\cal M})-Q^{(1)}({\cal M}) \bigr].
\ee
\end{definition}

By Lemma~\ref{lemma:chain}, we have
\be
  \label{q-order}
  Q^{(1)}({\cal N}) \le Q({\cal N}) \le Q_p({\cal N}) \le Q^{(1)}_p({\cal N}).
\ee

The \emph{symmetric side-channel assisted quantum capacity},
$Q_{ss}$, introduced and investigated in~\cite{Smith}, is obtained
by restricting the above optimization to channels ${\cal M}$ that
are symmetric, i.e.~both degradable and anti-degradable, which is a 
special subclass of zero-capacity channels. Unlike $Q$, $Q_{ss}$ is 
additive and has many other nice properties, and from the definition
and the above, we have (cf.~\cite[Sec.~VII]{Smith})
\be
  Q_{ss}({\cal N}) \le Q_p({\cal N}) \le Q^{(1)}_p({\cal N}).
\ee

How do we establish the upper bound for the potential quantum capacity? 
The idea is channel simulation inspired by the approach to obtain an upper 
bound for the quantum capacity: If the channel ${\cal N}$ can be simulated 
by another channel ${\cal N}^{\uparrow}$ using pre- and post-processing,
i.e.~${\cal N} = {\cal T} \circ {\cal N}^\uparrow \circ {\cal S}$ with
suitable CPTP maps ${\cal S}$ and ${\cal T}$, 
then clearly $Q({\cal N})\le Q({\cal N}^{\uparrow})$. We call ${\cal N}^\uparrow$
a \emph{lifting} of ${\cal N}$.
Furthermore, if the channel ${\cal N}^{\uparrow}$ is degradable, then its 
quantum capacity is given by the single-letter capacity $Q^{(1)}({\cal N}^{\uparrow})$, 
and obtain a single-letter upper bound for $Q({\cal N})$. This was observed
and exploited before under the name of ``additive extensions''~\cite{SmithSmolin-lifting}.

From inequality (\ref{q-order}) and the definition of potential single-letter
quantum capacity, we see that we should try to lift the channel to a strongly
additive one, because then we get even an upper bound for the potential quantum 
capacity, and in fact the potential single-letter quantum capacity! 

However it is not enough to lift the channel to a degradable one, because we learn 
from the superactivation phenomenon that its single-letter quantum capacity is not 
strongly additive. But an even narrower class of degradable channels, called 
Hadamard channels, satisfies the required property. 

\begin{definition}
  A \emph{Hadamard channel (HC)}~\cite{Hadamard} ${\cal N}$ is a quantum channel
  whose complementary channel ${\cal N}^c$ is an \emph{entanglement-breaking channel (EBC)} \cite{EBC}, where ${\cal N}^c$ can be expressed as
  \[
    {\cal N}^c(\rho) = \sum_i \proj{\phi_i} \bra{\tilde{\psi_i}}\rho \ket{\tilde{\psi_i}},
  \]
  in which $\sum_i \proj{\tilde{\psi_i}} = \1$ is a POVM. 
  Such channels are said to be entanglement-breaking because the output state 
  $\id\otimes{\cal N}^c(\rho^{RA})$ is separable for any state $\rho^{RA}$.
\end{definition}

The isometry of the Hadamard channel ${\cal N}$ is of the form 
(up to local unitary operation on $E$)
\be
  \label{eq:iso-H}
  V = \sum_i \ket{i}^B \ket{\phi_i}^E \bra{\tilde{\psi_i}}^A,
\ee
from which we see that the Hadamard channel ${\cal N}$ can simulate its 
complementary channel ${\cal N}^c$ by the operation of first measuring in the 
basis $\ket{i}$ and then preparing the state $\ket{\phi_i}$ 
according to the outcome of the measurement. 
Thus Hadamard channels are special degradable channels. 
\begin{proposition}[Cf. Bradler et al.~\cite{Bradler}, Wilde/Hsieh~\cite{WildeHsieh}]
\label{prop:H-strong-add}
If ${\cal N}$ is a Hadamard channel, then $Q^{(1)}$ is strongly additive:
$Q^{(1)}({\cal N}\otimes{\cal M}) = Q^{(1)}({\cal N})+Q^{(1)}({\cal M})$ 
for any contextual channel ${\cal M}$.
\end{proposition}

\begin{proof} 
The ``$\ge$'' part is trivial and we only need to prove the ``$\le$'' part.
Suppose the isometry of the Hadamard channel ${\cal N}$ is 
$V: A_1\hookto B_1\otimes E_1$, of the form (\ref{eq:iso-H}), 
and the isometry of ${\cal M}$ is $W: A_2\hookto B_2\otimes E_2$. 
The output, for an input state $\ket{\phi}^{RA_1A_2}$, is 
\[\begin{split}
  V\otimes W \ket{\phi}^{RA_1A_2}
         &= V \ket{\phi}^{RA_1B_2E_2},                                                                 \\
         &= \sum_i \ket{i}^{B_1}\ket{\phi_i}^{E_1} \bra{\tilde{\psi_i}}^{A_1}\ket{\phi}^{RA_1B_2E_2}, \\
         &= \sum_i \sqrt{p_i} \ket{i}^{B_1} \ket{\phi_i}^{E_1} \ket{\psi_i}^{RB_2E_2}.
\end{split}\]
The coherent information is thus,
\[\begin{split}
  &\quad~S(B_1B_2)-S(E_1E_2)\\
   &=   S(B_1)+S(B_2|B_1)-S(E_1)-S(E_2|E_1),\\
                      &\le S(B_1)-S(E_1)+S(B_2|Y)-S(E_2|Y),\\
                      &= S(B_1)-S(E_1)+\sum_ip_i(S(B_2)_{\psi_i}-S(E_2)_{\psi_i}),\\
                      &\le Q^{(1)}({\cal N})+Q^{(1)}({\cal M}),
\end{split}\]
where we use the isometry $V:\ket{i}^{B_1} \mapsto \ket{i}^{Y}\ket{i}^{Z}$ and the fact that $S(A|B) \le S(A|C)$ if there is an operation 
${\cal E}^{B\to C}$ satisfying $\rho_{AC}=\id\otimes{\cal E}^{B\to C}(\rho_{AB})$: $S(B_2|B_1)\le S(B_2|Y)$ from the operation ${\cal E}^{B_1\to Y}(\rho)=\tr_EV\rho V^{\dagger}$ and $S(E_2|E_1)\ge S(E_2|Y)$ from the operation ${\cal E}^{Y\to E_1}(\rho)=\sum_i\bra{i}\rho\ket{i}^Y\ketbra{\phi_i}{\phi_i}^{E_1}$.
\end{proof}

\medskip
\begin{remark}
\normalfont
The above proposition~\ref{prop:H-strong-add} is a special case of a more general 
fact \cite[Lemma 4]{Bradler}, \cite[Thm.~3 \&{} Lemma 2]{WildeHsieh}; cf.~also~\cite[App.~B, Lemma 7]{Devetak-degradable} 
for the special case of ``generalized dephasing channels'' (also known as Schur multipliers).
\end{remark}

\medskip
An immediate corollary is the following.
\begin{corollary}
  The potential quantum capacity (and potential single-letter quantum capacity)
  of a Hadamard channel ${\cal N}$ is equal to its single-letter quantum capacity:\\
  \[
    Q^{(1)}_p({\cal N}) = Q_p({\cal N}) = Q({\cal N}) = Q^{(1)}({\cal N}).
  \]
\end{corollary}

Thus we have reduced our task to finding a good lifting of a given
channel to a Hadamard channel. The question how to find the optimal one
is of interest in itself and we will discuss the general method elsewhere \cite{Yang2015}. 
For our present purposes, there is a rather straightforward way to lift
a channel to a Hadamard channel. Namely, choose Kraus operators for ${\cal N}$
as ${\cal N}(\rho)=\sum K_i\rho K_i^{\dagger}$. Then a Stinespring isometry 
for ${\cal N}$ can be written as $U = \sum_i \ket{i}^E K_i^{A\rightarrow B}$. 

Let us define a new channel, the lifting ${\cal N}^\uparrow$, via its
isometry
\[
  V^{A\rightarrow BB'\otimes E} := \sum_i \ket{i}^E \ket{i}^{B'} K_i^{A\rightarrow B},
\]
where the environment system is still $E$, but the receiver has now $BB'$,
and $B'$ holds a coherent copy of $E$. As we now give a copy of $E$ to
the channel output, the output of the complementary channel of ${\cal N}^\uparrow$
will be completely decohered in the $\{\ket{i}\}$ basis, so the complementary
channel is EBC, hence ${\cal N}^\uparrow$ is Hadamard, as desired.
The Kraus operators of the lifted Hadamard channel ${\cal N}^{\uparrow}$ are
$\{\ket{i} \otimes K_i\}$, and one way to write the channel is as
\[
  {\cal N}^{\uparrow}(\rho) = \sum_i \proj{i}\otimes K_i\rho K_i^{\dagger},
\]
 which we call the \emph{canonical lifting}.
Its quantum capacity is
\begin{align*}
  Q({\cal N}^\uparrow) &= \max_{\ket{\phi}^{RA}} \sum_i p_i S(\rho_i^B),\\
      \text{s.t. } p_i &= \tr (\1 \otimes K_i) \proj{\phi} (\1\otimes K_i)^{\dagger},\\
                \rho_i &= \frac{1}{p_i} \tr_R (\1 \otimes K_i) \proj{\phi} (\1\otimes K_i)^{\dagger}.
\end{align*}
Now take the minimum over all different Kraus representations (which after all
we are free to choose), to obtain the best bound from this particular family
of canonical liftings.

As a result, the quantum capacity of the optimal canonical lifting is
equal to the entanglement of formation of the original channel, which is 
defined as
\[
  E_F({\cal N}) := \max_{\ket{\phi}^{RA}} \min_{\{K_i\}} \sum_i p_i E(\phi_i),
  {\sqrt{p_i}} \ket{\phi_i}^{RB} = (\1\otimes K_i)\ket{\phi},
\]
and where $E(\varphi)$ is the entropy of entanglement of the
bipartite pure state $\varphi$.

The following lemma is implied by the proof of \cite[Lemma~13]{Berta}, 
though not explicitly stated there. 

\begin{lemma}[Berta \emph{et al.}~\cite{Berta}] 
  \label{minmax} 
  With the above notation, the following minimax formula holds:
  \[
    \min_{\{K_i\}} \max_{\ket{\phi}^{RA}} \sum_i p_i E(\phi_i)
            = \max_{\ket{\phi}^{RA}} \min_{\{K_i\}} \sum_i p_i E(\phi_i),
  \]
  where the infimum is taken over all Kraus representations of the channel ${\cal N}$.
\end{lemma}

\medskip
Now we obtain an upper bound on potential quantum capacity in terms of 
the entanglement of formation of the channel:

\begin{theorem}
  \label{thm:upper-Q_p}
  For a general channel ${\cal N}$, we have the following upper bound
  on the potential quantum capacity:
  \[
    Q_p({\cal N}) \le Q^{(1)}_p({\cal N}) \le E_F({\cal N}).
  \]
\end{theorem}

\begin{proof} 
Lifting the channel to the optimal canonical Hadamard channel, we get
\[\begin{split}
Q^{(1)}({\cal N}\otimes{\cal M}) 
   &\le Q^{(1)}({\cal N}^{\uparrow}\otimes{\cal M}),    \\
   =  &Q^{(1)}({\cal N}^{\uparrow})+Q^{(1)}({\cal M})=   E_F({\cal N})+Q^{(1)}({\cal M}),
\end{split}\]
where the first inequality comes from simulation, the first equality from 
the strong additivity (Proposition~\ref{prop:H-strong-add}), 
and the second equality from Lemma \ref{minmax}.
\end{proof}

Analogous to the classical capacity, we call a channel perfect if its quantum capacity is equal to $\log d_{\min}$ with $d_{\min} =\min\{d_{in}, d_{out}\}$.
Before we have the similar corollary, we recall a result in \cite{SW-QEC}.

\begin{lemma}[Schumacher/Westmoreland~\cite{SW-QEC}]
\label{lemma:SW-QEC}
For a pure bipartite state $\ket{\phi^{RA}}$, let the system $A$ be transmitted through a 
channel ${\cal N}: A\to B$ and the joint output state is $\rho^{RB}=\id\otimes {\cal N}(\phi^{RA})$. 
If $E_F(R:B)=S(R)$, then there exists a quantum operation ${\cal D}:B\to A$ such that 
$\id\otimes {\cal D}(\rho^{RB})=\phi^{RA}$.
\end{lemma}

\begin{corollary}
If a quantum channel ${\cal N}$ is not perfect for transmitting quantum 
information, then its potential single-letter quantum capacity is not maximal, either:
\[
  Q^{(1)}({\cal N}) < \log d_{\min} \ \Longrightarrow\  Q^{(1)}_p({\cal N}) < \log d_{\min}.
\]
In particular, if a quantum channel is not perfect for transmitting 
quantum information, then it cannot be activated to a perfect one by 
any zero-quantum-capacity channels.
\end{corollary}

\begin{proof} 
Suppose that the potential quantum capacity is $\log d_{\min}$, then from the upper 
bound in Theorem~\ref{thm:upper-Q_p} we know that $E_F({\cal N})=\log d_{\min}$. 

In the case of $d_{\min}=d_{in}$, from Lemma~\ref{lemma:SW-QEC},  we know that the channel operation can be perfectly corrected by a suitable operation ${\cal D}$ acting on B that means the channel is already perfect.  

In the case of $d_{\min}=d_{out}$, denoting the isometry of the channel $\cN:A\rightarrow B$ by $U:A \hookrightarrow B\otimes E$, there exists hence an input $\ket{\phi}^{RA}$ 
such that the output state plus the environment is 
$\ket{\phi}^{RBE}=U \ket{\phi}^{RA}$, where $\rho^{RB}=\id\otimes {\cal N}(\phi_{RA})$ 
satisfying $E_F(\rho^{RB}) = S(B) = \log d$. This implies that $\rho^{BE}$ is a product 
state. From the Uhlmann theorem \cite{Uhlmann}, 
we know there is a unitary $V:R\rightarrow R_1\otimes R_2$ such that 
\[\begin{split}
  V \ket{\phi}^{RBE} = V \otimes U \ket{\phi}^{RA}
          &=   \ket{\phi}_1^{R_1B} \otimes \ket{\phi}_2^{R_2E},\\
          &=   \sum_{i,j} \frac{1}{\sqrt{d}} \lambda_j \ket{i}^{R_1}\ket{e_i}^B\ket{j}^{R_2}\ket{e_j}^E,
\end{split}\]
where $\{\ket{i}^{R_1},\ket{e_i}^B\}$ are the bases in the Schmidt decomposition of $\ket{\phi}_1^{R_1B}$, and $\{\ket{j}^{R_2},\ket{e_j}^E\}$ for $\ket{\phi}_2^{R_2E}$.
So the input state $V \otimes \1 \ket{\phi}^{RA}=\sum_{i} \frac{1}{\sqrt{d}} \ket{i}^{R_1}\ket{j}^{R_2} U^{\dagger}(\ket{e_i}^B\ket{e_j}^{E})$ will give a product state
as output. Now from this input we construct a new input 
\[
  \ket{\psi}^{RA}
    = \sum_{i} \frac{1}{\sqrt{d}} \ket{i}^{R_1}\ket{0}^{R_2} U^{\dagger}(\ket{e_i}^B\ket{e_0}^{E}),
\]
yielding the output
\[
  U \ket{\psi}^{RA} = \sum_{i} \frac{1}{\sqrt{d}} \ket{i}^{R_1}\ket{0}^{R_2}\ket{e_i}^B\ket{e_0}^{E},
\]
which is the desired output of product state between $B$ and $E$.
This means $Q^{(1)}({\cal N})=\log d$, i.e.~${\cal N}$ is noiseless already.
\end{proof}

\begin{remark}
\normalfont 
We know that the channel can be very entangled even though its quantum capacity 
is zero. It is very difficult to characterize channels with zero quantum capacity.
So it seems that it is hard to say whether a noisy channel can be activated 
into a noiseless one under the assistance of zero-quantum-capapcity channels. 
However from the notion of potential quantum capacity, we can answer this question
in the negative.
\end{remark}

\subsection{Potential private capacity}
\label{subsec:P}
In this section, we repeat the analysis of the preceding subsection,
but for the potential private capacity. We shall show that, as for $Q$,
the potential private capacity cannot be maximal without it already 
being the single-letter private capacity.
Especially we prove that the private capacity of Hadamard channels is strongly
additive.

\begin{definition}
Specializing Definition~\ref{def:f-pot}
to the case $f\equiv P$, we obtain the \emph{potential private capacity}
\be
  P_p({\cal N}) = \sup_{\cal M} \bigl[ P({\cal N}\otimes{\cal M})-P({\cal M}) \bigr],
\ee
and the \emph{potential single-letter private capacity}
\be
  P^{(1)}_p({\cal N}) = \sup_{\cal M} \bigl[ P^{(1)}({\cal N}\otimes{\cal M})-P^{(1)}({\cal M}) \bigr].
\ee
\end{definition}

By Lemma~\ref{lemma:chain}, we have
\be
  \label{p-order}
  P^{(1)}({\cal N}) \le P({\cal N}) \le P_p({\cal N}) \le P^{(1)}_p({\cal N}).
\ee

As for in the quantum case, we aim to lift channels to strongly 
additive ones, so as to obtain single-letter upper bounds on the
potential private capacity. Indeed, we can extend Proposition~\ref{prop:H-strong-add}
to the private capacity:

\begin{proposition}
  If ${\cal N}$ is a Hadamard channel, then 
  $P^{(1)}$ is strongly additive:
  $P^{(1)}({\cal N}\otimes{\cal M})=P^{(1)}({\cal N})+P^{(1)}({\cal M})$ 
  for any contextual channel ${\cal M}$.
\end{proposition}

\begin{proof} 
The ``$\ge$'' part is trivial and we only need to prove the ``$\le$'' part.
Suppose the isometry of the Hadamard channel 
${\cal N}$ is $V: A_1\hookto B_1\otimes E_1$, of the form (\ref{eq:iso-H}) 
as before, 
and the isometry of ${\cal M}$ is $W: A_2\hookto B_2\otimes E_2$. For the input state ensemble $\{p_t,\rho_t^{A_1A_2}\}$, we construct the classical-quantum (cq) state with the reference system R that purifies each $\rho_t^{A_1A_2}$,  
\[
  \sum_{t} p_t \proj{t}^{T} \otimes \proj{\phi_t}^{RA_1A_2},
\]
which is mapped by $V \otimes W$ to
\[
  \sum_{t} p_{t} \proj{t}^{T} \otimes \proj{\phi_t}^{RB_1E_1B_2E_2}.
\]
Here,
\be
\label{eq:state}
  \ket{\phi_t}^{RB_1E_1B_2E_2}
      = \sum_i \sqrt{q_{i|t}} \ket{i}^{B_1} \ket{\phi_i}^{E_1} \ket{\psi_{i|t}}^{RB_2E_2}.
\ee

Using the isometry $\ket{i}^{B_1} \mapsto \ket{i}^{Y} \ket{i}^{Z}$, and the notation $F(X)|T=\sum p_tF(X_t)$, we get
\[\begin{split}
 &\quad~I(T:B_1B_2) -I(T:E_1E_2)\\
               &=   S(B_1B_2)-S(E_1E_2)-(S(TB_1B_2)-S(TE_1E_2)),\\
              &=   S(B_1)-S(E_1)+S(B_2|B_1)-S(E_2|E_1)\\
              &~-(S(B_1B_2|T)-S(E_1E_2|T)),\\
              &\le S(B_1)-S(E_1)+S(B_2|Y)-S(E_2|Y)\\
              &~-(S(B_1B_2|T)-S(E_1E_2|T)),\\
              &=   S(B_1)-S(E_1)+[I(T:B_2)-I(T:E_2)]|Y\\
              &~+(S(TB_2|Y)-S(TE_2|Y))-(S(B_1B_2|T)-S(E_1E_2|T)),\\
              &=   S(B_1)-S(E_1)+[I(T:B_2)-I(T:E_2)]|Y\\
              &~+(S(B_2Y|T)-S(E_2Y|T))-(S(B_1B_2|T)-S(E_1E_2|T)),\\
              &=   S(B_1)-S(E_1)+[I(T:B_2)-I(T:E_2)]|Y\\
              &~+[S(B_2Y)-S(E_2Y)-S(B_1B_2)+S(E_1E_2)]|T,\\
              &\le P^{(1)}({\cal N})+P^{(1)}({\cal M})\\
              &~+[S(B_2Y)-S(E_2Y)-S(B_1B_2)+S(E_1E_2)]|T,
\end{split}\]
where the first inequality comes from $S(B_2|B_1)\le S(B_2|Y)$ and $S(E_2|E_1)\ge S(E_2|Y)$, which hold for the same reason in the proof of Prop. \ref{prop:H-strong-add}.
Next we show that each term in the average over $T$ is non-positive.
Indeed, we evaluate the term in the pure state of the form (\ref{eq:state}) and notice that $S(E_2Y)=S(E_1E_2Y)$. Then we have
\[\begin{split}
  &\quad~ S(B_2Y)-S(E_2Y)-S(B_1B_2)+S(E_1E_2)\\
  & = -I(R:Y|E_1E_2)\leq 0,
\end{split}\]
which concludes the proof.
\end{proof}

\medskip
An immediate corollary is the following.
\begin{corollary}
  \label{cor:Hadamard-P_p}
  The potential private capacity (and the potential single-letter private capacity)
  of a Hadamard channel ${\cal N}$ is equal to its private capacity, which
  in turn equals $Q^{(1)}({\cal N})$:
  \[
    P^{(1)}_p({\cal N}) = P_p({\cal N}) = P({\cal N}) = Q^{(1)}({\cal N}).
  \]
\end{corollary}

\begin{theorem} 
  For any channel ${\cal N}$, we have the upper bound
  $P_p({\cal N}) \le P_p^{(1)}({\cal N}) \le E_F({\cal N})$.
\end{theorem}

\begin{proof} 
Notice that for the Hadamard channel, the potential capacity is equal 
to the quantum capacity and the rest is the same as the quantum case,
i.e.~the proof of Theorem~\ref{thm:upper-Q_p}.
\end{proof}

\medskip
As in the quantum case, we obtain the following immediate corollary:

\begin{corollary}
If a quantum channel is not perfect for transmitting private 
information, then it cannot be activated to the perfect one by any contextual channels.
\end{corollary}


\section{Discussion and open questions}
\label{sec:open}
\IEEEPARstart{W}{e} have introduced potential capacities of quantum channels as ``big brothers" of the plain capacities, to capture the degree of non-additivity of the latter in the most favorable context. By bootstrapping strongly additive channels, i.e. those whose plain capacity equals its potential version, we were able to give some general upper bounds on various potential capacities. While the potential concept makes sense for any capacity, here we focused on a few examples,
 basically the principal channel capacities C, Q and P. Our central result is that a noisy channel cannot be activated into a noiseless one by any 
contextual channel. This result holds for the classical, quantum, 
and private capacity, and improves upon previous statements. 
Notice that in the notion of potential capacity, a PPT-entanglement-binding 
channel may have positive potential quantum capacity. So it is tempting to speculate
whether all entangled channels have positive potential quantum capacity.
This is a big open question deserving of study in the future. 

Looking beyond capacities and at the tradeoff between different
resources, note that Hadamard channels that served us so
well in the treatment of the potential quantum and classical
capacities, also allow for a single-letter formula for the 
qubit-cbit-ebit tradeoff region~\cite{WildeHsieh}; in fact, Thm.~3
and Lemma 2 of that paper show that this region is strongly
additive regarding the tensor product of a Hadamard channel
with any other channel. Thus, the notion of lifting to a Hadamard
channel once more yields and outer bound on the potential
capacity region of the achievable triples $(q,c,e)$.

We have studied potential capacities only for the basic quantities, 
and one ($Q_A$) for which we could calculate the potential
capacity exactly. For most capacities, we may assume that it will be
prohibitive to calculate the potential version as well as
its plain version, so we have to be content with bounds. In the 
domain of zero-error information theory, other exact characterizations of 
some potential capacities are known~\cite{ADSW}.

In fact, in~\cite{ADSW}, differences between the capacity of
$K(\cN\ox\cM)$ and \emph{another} parameter for $\cM$, which
represents a more general value $V(\cM)$ of the channel, were
considered, where $K$ is superadditive and $V(\cM) \geq K(\cM)$.
Then, 
\[
  V^*(\cN) = \sup_{\cM} K(\cN\ox\cM) - V(\cM)
\]
is a kind of amortized value of $\cN$, the rationale being that
the gain from the ``profit'' $K(\cN\ox\cM)$ has to be offset by
the ``price'' $V(\cM)$ of the borrowed resource. Of particular
interest is the case of an economically fair pricing $V=V^*$, 
i.e.~of a situation where the same value $V(\cN)=V^*(\cN)$ quantifies the price
of resource when we have to borrow it, as well as the amortized
value in a suitable context. In the setting of zero-error
capacities, specifically $K=\log\alpha$ (with the independence number
$\alpha$), this has been shown to hold true for $V=\log\vartheta$
(the Lov\'{a}sz number). 
For Shannon theoretic capacities, i.e.~$K\in\{C,P,Q\}$ or similar, the 
existence and possible characterization of a value $V$ that yields a 
fair value $V=V^*$ is perhaps one of the most intriguing questions 
raised by the notion of potential capacities. For instance, for the
quantum capacity, it turns out that its symmetric side-channel assisted
version $Q_{ss}$ is such a fair price. Namely, for a given channel $\cN$ and 
any contextual channel $\cM$,
\[\begin{split}
  Q^{(1)}(\cN\ox\cM) &\leq Q(\cN\ox\cM),\\
  & \leq Q_{ss}(\cN\ox\cM) = Q_{ss}(\cN)+Q_{ss}(\cM),
\end{split}\]
thus
\[\begin{split}
  &\quad~\sup_{\cM} Q^{(1)}(\cN\ox\cM) - Q_{ss}(\cM) \\
  &\leq \sup_{\cM} Q(\cN\ox\cM) - Q_{ss}(\cM) \leq Q_{ss}(\cN).
\end{split}\]
On the other hand, restricting to symmetric channels $\cM$ in this optimization,
for which $Q_{ss}(\cM) = 0$, we attain equality asymptotically by definition of
the symmetric side-channel assisted quantum capacity. The same reasoning 
can be applied to the private capacity and the symmetric side-channel
assisted version $P_{ss}$~\cite{Smith-degradable-p}, so we have proved the
following.

\begin{theorem}
  For any channel $\cN$, 
  \begin{align*}
    Q_{ss}(\cN) &= \sup_{\cM} Q^{(1)}(\cN\ox\cM) - Q_{ss}(\cM), \\
    &= \sup_{\cM} Q(\cN\ox\cM) - Q_{ss}(\cM), \\
    P_{ss}(\cN) &= \sup_{\cM} P^{(1)}(\cN\ox\cM) - P_{ss}(\cM),\\
    &  = \sup_{\cM} P(\cN\ox\cM) - P_{ss}(\cM),
  \end{align*}
  and equality is asymptotically attained by symmetric channels $\cM$.
\end{theorem}

If the quantum and private capacities have alternative fair pricing schemes
along these lines or if $Q_{ss}$ ($P_{ss}$) are unique, and whether there is an 
analogous statement for $C$ or $\chi$, remain open.


\section*{Acknowledgments}
We thank Ke Li, Graeme Smith and John Smolin for helpful 
discussions, and Runyao Duan, Marius Junge and Mark Wilde 
for interesting remarks, on the concept of potential
capacity and potential capacity regions, respectively.
Part of this work was done during the programme 
\emph{Mathematical Challenges in Quantum Information (MQI)} at the 
Isaac Newton Institute in Cambridge, whose hospitality was gratefully 
acknowledged, and where DY was supported by a Microsoft Visiting Fellowship. 
DY's work is supported by the ERC (Advanced Grant ``IRQUAT'') and the NSFC (Grant No. 11375165). 
AW's work is supported by the European Commission (STREP ``RAQUEL''), the European Research 
Council (Advanced Grant ``IRQUAT''), the Spanish MINECO (projects FIS2008-01236
and FIS2013-40627-P), with the support of FEDER funds, as well as by
the Generalitat de Catalunya CIRIT, project 2014-SGR-966.


\section*{Appendix}
Here we analyze the structure of the state satisfying $S(B)-S(BE)=G(B:E)$ or 
$S(B)-S(BE)=E_F(B:E)$. The general relation among these three quantities is 
$S(B)-S(BE)\le G(B:E)\le E_F(B:E)$, the first ``$\le$'' comes from Lemma~\ref{lemma:IG} 
and the second ``$\le$'' from Lemma~\ref{lemma:yang}. Obviously the condition 
$S(B)-S(BE)=E_F(B:E)$ implies $S(B)-S(BE)=G(B:E)$. Lemma~\ref{lemma:IGE} asserts 
that the latter also implies the former. 

\begin{lemma}
\label{lemma:IG}
For a mixed state $\rho^{BE}$, $S(B)-S(BE)\le C_{\leftarrow}(\rho^{BE})$, 
and equality holds iff there exists a unitary on $B$ such that 
$U_B\rho^{BE}U_B^{\dagger}=\rho^{B_L}\otimes\phi^{B_RE}$, where 
${\cal H}_B={\cal H}_{B_L}\otimes {\cal H}_{B_R}$ and $\phi^{B_RE}$ is pure.
\end{lemma}

\begin{proof}
Consider the purification $\phi^{RBE}$ of the state $\rho^{BE}$. 
Then $C_{\leftarrow}(\rho^{BE})=S(B)-E_F(R:B)$ and $S(BE)=S(R)$. 
From the inequality $E_F(R:B)\le S(R)$, we arrive at 
$S(B)-S(BE)\le C_{\leftarrow}(\rho^{BE})$.

When the equality holds, this amounts to $E_F(R:B)=S(R)$. From the relation 
$C_{\leftarrow}(\rho^{RE})=S(R)-E_F(R:B)=0$, we get that $\rho^{RE}=\rho^R\otimes\rho^E$. 
From Uhlmann's theorem \cite{Uhlmann}, there exists a unitary $U_B$ such that 
$U_B\phi^{RBE}U_B^{\dagger}=\phi^{RB_L}\otimes\phi^{B_RE}$. Tracing out $R$ 
concludes the proof.
\end{proof}

\begin{lemma}
\label{lemma:IGE}
For a state $\rho^{BE}$, $S(B)-S(BE)\le G(B:E)$. If $S(B)-S(BE)=G(B:E)$, then $S(B)-S(BE)=E_F(B:E)$.
\end{lemma}

\begin{proof}
Suppose that the optimal realization of $G(B:E)$ is the state ensemble $\{p_i,\rho_i^{BE}\}$.
Then,
\[
S(B)-S(BE) \le  \sum_i p_i[S(B_i)-S(BE_i)]\le \sum_i p_i C_{\leftarrow}(\rho^{BE}_i),
\]
where the first ``$\le$'' comes from the concavity of the conditional entropy and 
the second ``$\le$'' from Lemma~\ref{lemma:IG} above.

If  $S(B)-S(BE)=G(BE)$, then $S(B_i)-S(BE_i)=C_{\leftarrow}(BE_i)$ for each $\rho_i^{BE}$. From Lemma~\ref{lemma:IG}, the state $\rho_i^{BE}$ has the property $E_F(BE_i)=S(B_i)-S(BE_i)$. Then $S(B)-S(BE)\le E_F(BE)\le \sum_i p_i E_F(BE_i)=\sum_i p_i [S(B_i)-S(BE_i)]=S(B)-S(BE)$ and the proof ends.
\end{proof}

~\\
In fact, the constraint is so sharp that we can even learn the structure of the bipartite state. Proposition \ref{prop:structure} explains this and is also of independent interest. 

\begin{proposition}
\label{prop:structure}
A state $\rho_{BE}$ in the finite dimensional Hilbert space ${\cal H}_B\otimes {\cal H}_E$ 
satisfies $S(B)-S(BE)=E_F(BE)$, if and only if it is of the form 
\[
  \rho^{BE} = \bigoplus_i p_i\rho_{i}^{B_i^L}\otimes\phi_{i}^{B_i^RE}, 
\]
where $\phi_{i}^{B_i^RE}$ are pure states and the system $B$ is decomposed into 
the direct sum of tensor products
\[
  {\cal H}_{B} = \bigoplus_i {\cal H}_{B_i^L}\otimes{\cal H}_{B_i^R}.
\]
\end{proposition}
\begin{proof}
Suppose the optimal realization of $E_F(B:E)$ is the ensemble $\{p_x,\ket{\psi_x}\}$,
and construct the state $\rho^{XBE}=\sum_x p_x \proj{x} \otimes \psi_x^{BE}$. 
From the condition $S(B)-S(BE)=E_F(B:E)$, we get $S(B)-S(BE)=\sum_x p_x S(\psi^B_x)$. 
This condition can be expressed as $I(X:E|B)=0$, where $I(X:E|B)=S(XB)+S(EB)-S(XBE)-S(B)$ 
is the conditional quantum mutual information. From \cite{Hayden-equality}, we know 
that $I(X:E|B)=0$ if and only if the state $\rho^{XBE}$ can be decomposed as
\[
  \rho^{XBE} = \bigoplus_i q_i \rho_i^{XB_i^L} \otimes \rho_i^{B_i^RE},
\]
where the system $B$ is decomposed into the direct sum of tensor products
\[
  {\cal H}_{B} = \bigoplus_i {\cal H}_{B_i^L}\otimes{\cal H}_{B_i^R}.
\] 

Thus we have
\[
  \rho^{BE} = \bigoplus_i q_i\rho_i^{B_i^L}\otimes\rho_i^{B_i^RE}.
\]
If all the states $\rho_i^{B_i^RE}$ are pure, then we are done. In general, 
some of $\rho_i^{B_i^RE}$ may be mixed. Apply the condition $S(B)-S(BE)=E_F(B:E)$ 
to the structured state $\rho^{BE}$, we get $\sum q_i(S(B_i^R)-S(B_i^RE))=\sum q_iE_F(B_i^R:E)$.  
Since $S(B)-S(BE)\le E_F(B:E)$ is true for all of the components
$\rho_i^{B_i^RE}$, we arrive at $S(B_i^R)-S(B_i^RE)=E_F(B_i^R:E)$ 
for each $i$.
So we can use the argument again and get that the structure of each state 
$\rho_i^{B_i^R:E}$ is of the direct sum of tensor products. If some of the 
new states $\rho_{i(j)}^{B_{i(j)}^{R}E}$ are mixed, we repeat the argument 
for these states. In each iteration, the dimension is reduced because of 
the direct sum of tensor products. Since system $B$ is a finite dimensional Hilbert 
space, the iteration ends after finitely many steps when all the states 
$\rho_{i(\cdots)}^{B_{i(\cdots)}^{R}E}$ are pure. After renumbering the 
labels, we get the desired decomposition where all the states $\phi_{i}^{B_i^RE}$ are pure. 
\end{proof}



\end{document}